\documentclass[12pt]{amsart}
\usepackage{a4wide}
\usepackage[dvips]{epsfig}
\usepackage{amscd}
\usepackage{amssymb}
\usepackage{amsthm}
\usepackage{amsmath}
\usepackage{latexsym}
\usepackage[linktocpage,
                  menucolor=blue,
                  linkcolor=blue,
                  citecolor=red,
                  colorlinks=true]{hyperref}

\theoremstyle{plain}
\newtheorem{theorem}{Theorem}[section]
\theoremstyle{plain}
\newtheorem{lemma}[theorem]{Lemma}

\theoremstyle{definition}
\newtheorem{definition}{Definition}[section]

\newtheorem{example}{Example}[section]

{%
\setcounter{enumi}{0}

\begin{enumerate}}%
{\end{enumerate} }
{%
\setcounter{enumi}{0}

\begin{enumerate}}%
{\end{enumerate} }

\numberwithin{equation}{section}
 \allowdisplaybreaks

\title[Risk-based optimization of an insurer]
 {Risk-based optimal portfolio of an insurer with regime switching and noisy memory }

\date{\today}

\begin{document}

\author{Rodwell Kufakunesu}

\address{Department of Mathematics and Applied Mathematics, University of
Pretoria, 0002, South Africa}
\email{rodwell.kufakunesu@up.ac.za}

\author{Calisto Guambe}
\address{Department of Mathematics and Applied Mathematics, University of
Pretoria, 0002, South Africa}
\address{Department of Mathematics and Informatics, Eduardo Mondlane University,
257, Mozambique}
\email{calistoguambe@yahoo.com.br}

\author{Lesedi Mabitsela}
\address{Department of Mathematics and Applied Mathematics, University of
Pretoria, 0002, South Africa}
\email{rodwell.kufakunesu@up.ac.za}


\begin{abstract}
In this paper, we consider a risk-based optimal investment problem of an insurer
in a regime-switching jump-diffusion model with noisy memory. Using the model
uncertainty modeling, we formulate the investment problem as a zero-sum,
stochastic differential delay game between the insurer and the market, with a
convex risk measure of the terminal surplus and the Brownian delay surplus over
a period $[T-\varrho,T]$. Then, by the BSDE approach, the game problem is
solved. Finally, we derive analytical solutions of the game problem, for a
particular case of a quadratic penalty function and a numerical example is considered.
\end{abstract}

\maketitle
\section{Introduction}

 Stochastic delay equations are equations whose coefficients depend also on past
history of the solution. They appear naturally in economics, life science,
finance, engineering, biology, etc. In Mathematics of Finance,
the basic assumption of the evolution price processes is that they are Markovian.
In reality, these processes possess some memory which cannot be neglected.
Stochastic delay control problems have received much interest in recent times and
these are solved by different methods. For instance, when the state process depends on the discrete and average delay, Elsanoni {\it
et. al.} \cite{elsanoni} studied an optimal harvesting problem using the dynamic
programming approach. On the other hand, a maximum principle approach was used
to solve optimal stochastic control systems with delay. See e.g., Oksendal and
Sulem \cite{oksendal2001}, Pamen \cite{pamen2015}. When the problem allows a noisy
memory, i.e., a delay modeled by a Brownian motion, Dahl {\it et. al}
\cite{Dahl} proposed a maximum principle approach with Malliavin derivatives to
solve their problem. For detailed information on the theory of stochastic delay
differential equations (SDDE) and their applications to stochastic control
problems, see, e.g., B\~anos {\it et. al.} \cite{banos}, Kuang \cite{kuang},
Mohammed \cite{mohammed} and references therein.

In this paper, we consider an insurer's risk-based optimal investment problem
with noisy memory. The financial market model set-up is composed by one risk-free asset and
one risky asset described by a hidden Markov regime-switching jump diffusion
process. The jump-diffusion models represent a valuable extension of the diffusion
models for modeling the asset prices \cite{oksendal2001}. They capture some sudden changes in the
market such as the existence of high-frequency data, volatility clusters and
regime switching. It is important to note that in the Markov regime-switching
diffusion models, we can have random coefficients possibly with jumps, even if
the return process is a diffusion one. In this paper, we consider a jump diffusion model,
which incorporates jumps in the asset price as well as in the model
coefficients, i.e., a Markov regime-switching jump-diffusion model. Furthermore, we consider
the Markov chain to represent different modes of the economic environment
such as, political situations, natural catastrophes, etc. Such kind of models
have been considered for option pricing of the contingent claim, see for
example, Elliott {\it et. al} \cite{elliott2007},  Siu \cite{siu2014} and
references therein. For stochastic optimal control problems, we mention the
works by B\"auerle and Rieder \cite{rieder}, Meng and Siu \cite{Meng}. In these works a
portfolio asset allocation and a risk-based asset allocation of a
Markov-modulated jump process model has been considered and solved via the
dynamic programming approach. We also mention a recent work by Pamen and Momeya
\cite{pamen2017}, where a maximum principle approach has been applied to an
optimization problem described by a Markov-modulated regime switching jump-diffusion model.

In this paper, we assume that the company receives premiums at the constant rate
and pays the aggregate claims modeled by a hidden Markov-modulated pure jump process.
We assume the existence of capital inflow or outflow from the insurer's current
wealth, where the amount of the capital is proportional to the
past performance of the insurer's wealth. Then, the surplus process is governed
by a stochastic delay differential equation with the delay, which may be random. Therefore we find it reasonable to consider  also a delay modeled by Brownian motion. In literature,
a mean-variance problem of an insurer was considered, but the wealth process is
given by a diffusion model with distributed delay, solved via the maximum principle
approach (Shen and Zeng \cite{Shen}). Chunxiang and Li \cite{chunxing2015} extended this mean-variance problem of an insurer to the Heston stochastic volatility case and solved using dynamic programming approach. For thorough discussion on different types of delay, we refer to Ba\~nos {\it et. al.} \cite{banos}, Section 2.2.

We adopt a convex risk measure first introduced by Frittelli and Gianin
\cite{frittelli} and F\"ollmer and Schied \cite{follmer}. This generalizes the
concept of coherent risk measure first introduced by Artzner {\it et. al.}
\cite{artzner}, since it includes the nonlinear dependence of the risk of the
portfolio due to the liquidity risks. Moreover, it relaxes a sub-additive and
positive homogeneous properties of the coherent risk measures and substitute
these by a convex property.

When the risky share price is described by a diffusion process and without
delay, such kind of risk-based optimization problems of an insurer have been
widely studied and reported in literature, see
e.g., Elliott and Siu \cite{Elliott2011,elliott}, Siu \cite{siu2012,siu,
siu2014}, Peng and Hu \cite{peng}. For a jump-diffusion case, we refer to
Mataramvura and \O ksendal \cite{Mataramvura}.

To solve our optimization problem, we first transform the unobservable Markov
regime-switching problem into one with complete observation by using the so-called filtering theory, where the optimal Markov chain is also derived. For
interested readers, we refer to Elliott {\it et. al.} \cite{elliott2008},
Elliott and Siu \cite{elliott}, Cohen and Elliott \cite{cohen} and Kallianpur
\cite{kallianpur}. Then we formulate a convex risk measure described by a
terminal surplus process as well as the dynamics of the noisy memory surplus
over a period $[T-\varrho,T]$ of the insurer to measure the risks. The main
objective of the insurer is to select the optimal investment strategy so as to
minimize the risk. This is a two-player zero-sum stochastic delayed differential
game problem. Using delayed backward stochastic differential equations (BSDE) with a
jump approach, we solve this game problem by an application of a comparison
principle for BSDE with jumps. Our modeling framework follows that in Elliott
and Siu \cite{Elliott2011}.

The rest of the paper is organized as follows: In Section 2, we introduce the
dynamic of state process described by SDDE in the Hidden Markov regime switching
jump-diffusion market. In Section 3, we use the filtering theory to turn the
model into one with complete observation. We also derive the optimal Markov
chain. Section 4, is devoted to the formulation of our risk-base optimization
problem as a zero-sum stochastic delayed differential game problem, which is
then solved in Section 5. Finally, in Section 6, we derive the explicit
solutions for a particular case of a quadratic penalty function and we give an example to show how one can apply these results in a concrete situation.

\section{Model formulation}
Suppose we have an insurer investing in a finite investment period $T<\infty$.
Consider a complete filtered probability space
($\Omega,\mathcal{F},\{\mathcal{F}_t\}_{0\leq t\leq T},\mathbb{P}$), where
$\{\mathcal{F}_t\}_{t\in[0,T]}$ is a filtration satisfying the usual conditions
(Protter \cite{Protter}).  Let $\Lambda(t)$ be a continuous time finite state hidden
Markov chain defined on $(\Omega, \mathcal{F}, \mathbb{P})$, with a finite
state space $\mathcal{S}=\{e_1,e_2,\ldots,e_D\}\subset\mathbb{R}_D$,
$e_j=(0,\ldots,1,0,\ldots,0)\in\mathbb{R}^D$, where $D\in\mathbb{N}$ is the
number of states of the chain, and the $j$th component of $e_n$ is the Kronecker
delta $\delta_{nj}$, for each $n,j=1,2,\ldots,D$.  $\Lambda(t)$ describes the
evolution of the unobserved state of the model parameters in the financial
market over time, i.e., a process which collects factors that are relevant for
the model, such as, political situations, laws or natural catastrophes (see,
e.g. Bauerle and Rieder \cite{rieder}, Elliott and Siu \cite{elliott}). The main
property of the Markov chain $\Lambda$ with the canonical state space $\mathcal{S}$ is that, any nonlinear function of
$\Lambda$, is linear in $\Lambda$, i.e.,
$\varphi(\Lambda)=\langle\varphi,\Lambda\rangle$, where
$\langle\cdot,\cdot\rangle$ denotes the inner product in $\mathbb{R}^D$. For
detailed information, see, for instance, Elliott {\it et. al.}
\cite{elliott2008}.

To describe the probability law of the chain $\Lambda$, we define a family of
intensity matrix $A(t):=\{a_{ji}(t);\,\,t\in[0,T]\}$, where $a_{ji}(t)$ is the
instantaneous transition intensity of the chain $\Lambda$ from state $e_i$ to
state $e_j$ at time $t\in[0,T]$. Then it was proved in Elliott {\it et. al.}
\cite{elliott2008}, that $\Lambda$ admits the following semi-martingale dynamics:

$$
\Lambda(t)=\Lambda(0)+\int_0^tA(s)\Lambda(s)ds+\Phi(t)\,,
$$
where $\Phi$ is an $\mathbb{R}^D$-valued martingale with respect to the natural filtration generated by $\Lambda$.

To describe the dynamics of the financial market, we consider a Brownian motion
$W(t)$ and a compensated Markov regime-switching Poisson random measure
$\tilde{N}_{\Lambda}(dt,dz):=N(dt,dz)-\nu_{\Lambda}(dz)dt$, with the dual
predictable projection $\nu_\Lambda$ defined by
$$
\nu_{\Lambda}(dt,dz)=\sum_{j=1}^D\langle\Lambda(t-),
e_j\rangle\varepsilon_j(t)\nu_j(dz)dt\,,
$$
where $\nu_j$ is the conditional Levy measure of the random jump size and
$\varepsilon_j$ is the intensity rate when the Markov chain $\Lambda$ is in
state $e_j\,$. We suppose that the processes $W$ and $N$ are independent.

We consider a financial market consisting of one risk-free asset $(B(t))_{0\leq
t\leq T}$ and one risky  asset $(S(t))_{0\leq t\leq T}$. Their respective prices
are given by the following regime-switching stochastic differential equations
(SDE):

\begin{eqnarray}\nonumber
  dB(t) &=& r(t)B(t)dt\,, \ \ B(0)=1\,, \label{risk-free} \\ \nonumber
  dS(t) &=&
S(t)\Bigl[\alpha^{\Lambda}(t)dt+\beta(t)dW(t)+\int_{\mathbb{R}}zN(dt,dz)\Bigl]
\\ \label{risky}
        &=&
S(t)\Bigl[\Bigl(\alpha^{\Lambda}(t)+\sum_{j=1}^D\int_{\mathbb{R}}
z\langle\Lambda(t-),e_j\rangle\varepsilon_j(t)\nu_j(dz) \Bigl) dt+\beta(t)dW(t)
+\int_{\mathbb{R}}z\tilde{N}_{\Lambda}(dt,dz)\Bigl]\,,
\end{eqnarray}
with initial value $S(0)=s>0$. We suppose that the instantaneous interest rate
$r(t)$ and the appreciation rate $\alpha(t)$ are modulated by the Markov chain
$\Lambda$, as follows:
\begin{eqnarray*}
  r(t) &:=& \langle{\bf
r}(t),\Lambda(t)\rangle=\sum_{j=1}^Dr_j(t)\langle\Lambda(t),e_j\rangle\,, \\
  \alpha^{\Lambda}(t) &:=& \langle
\alpha(t),\Lambda(t)\rangle=\sum_{j=1}^D\alpha_j(t)\langle\Lambda(t),e_j\rangle\
,,
\end{eqnarray*}
where $r_j$ and $\alpha_j$ represent the interest and appreciation rates
respectively, when the Markov chain is in state $e_j$ of the economy. We suppose
that ${\bf r}(t)$ and $\alpha(t)$ are $\mathbb{R}^D$-valued $\mathcal{F}_t$-
predictable and uniformly bounded processes on the probability space
$(\Omega,\mathcal{F},\mathbb{P})$. Otherwise,  the volatility rate $\beta(t)$
is an $\mathcal{F}_t$-adapted uniformly bounded process. Note that we may
consider a Markov modulated volatility process, however it would lead in a
complicated, if not possible filtering issue in the following section. As was
pointed out by Siu \cite{siu2014} and references therein, the other reason is
that, the volatility can be determined from a price path of the risky share,
i.e., the volatility is observable.

We now model the insurance risk by a Markov regime-switching pure jump process
on the probability space $(\Omega,\mathcal{F},\mathbb{P})$. We follow the
modeling framework of Elliott and Siu \cite{elliott}, Siu \cite{siu}, Pamen and
Momeya \cite{pamen2017}.

Consider a real valued pure jump process $Z:=\{Z(t);\,\,t\in[0,T]\}$ defined on a
probability space $(\Omega,\mathcal{F},\mathbb{P})$, where $Z$ denotes the
aggregate amount of the claims up to time $t$. Then, we can write $Z$ as
$$
Z(t)=\sum_{0<s\leq t}\Delta Z(s)\,; \ \ \ Z(0)=0, \ \ \mathbb{P}-\rm{a.s.}, \ \
t\in[0,T]\,,
$$
where $\Delta Z(s):=Z(s)-Z(s-)$, for each $s\in[0,T]$, represents the jump size
of $Z$ at time $s$.

Suppose that the state space of the claim size $\mathcal{Z}$ is $(0,\infty)$.
Consider a random measure $N^0(\cdot,\cdot)$ defined on a product space
$[0,T]\times\mathcal{Z}$, which selects the random claim arrivals at time $s$. The aggregate insurance claim process $Z$ can be
written as
$$
Z=\int_0^t\int_0^\infty zN^0(ds,dz); \ \ \ t\in[0,T]\,.
$$
Define, for each $t\in[0,T]$,
$$
M(t):=\int_0^t\int_0^\infty N^0(ds,dz); \ \ \ t\in[0,T]\,.
$$
$M(t)$ counts the number of claim arrivals up to time $t$. Suppose that under
$\mathbb{P}$, $M:=\{M(t),\,t\in[0,T]\}$ is a conditional Poisson process on
$(\Omega,\mathcal{F},\mathbb{P})$, given the information about the realized path of the chain, with intensity $\lambda^{\Lambda}(t)$
modulated by the Markov chain given by
$$
\lambda^{\Lambda}(t):=\langle\lambda(t),\Lambda(t)\rangle=\sum_{j=1}
^D\lambda_j\langle\Lambda(t),e_j\rangle\,,
$$
where $\lambda_j$ is the $j$th entry of the vector $\lambda$ and represents the
intensity rate of $M$ when the Markov chain is in the state space $e_j$.

Let $f_j(z)$, $j=1,\ldots,D$ be the probability density function of the chain
size $z=Z(s)-Z(s-)$, when $\Lambda(t-)=e_j$. Then the Markov regime-switching
compensator of the random measure $N^0(\cdot,\cdot)$ under $\mathbb{P}$, is
given by
$$
\nu^0_{\Lambda}(ds,dz):=\sum_{j=1}^D\langle\Lambda(s-),
e_j\rangle\lambda_j(s)f_j(dz)ds\,.
$$
Therefore, a compensated version of the random measure is given by
$$
\tilde{N}^0_{\Lambda}(ds,dz)=N^0(ds,dz)-\nu^0_{\Lambda}(ds,dz)\,.
$$

We suppose that $\tilde{N}^0_{\Lambda}$ is independent of $W$ and
$\tilde{N}_{\Lambda}$.

Let $p(t)$ be the premium rate at time $t$. We suppose that the premium rate
process $\{p(t),\, t\in[0,T]\}$ is $\mathcal{F}_t$-progressively measurable and uniformly
bounded process on $(\Omega,\mathcal{F},\mathbb{P})$, taking values on
$(0,\infty)$. Let $R:=\{R(t),\,t\in[0,T]\}$ be the insurance risk process of the
insurance company without investment. Then, $R(t)$ is given by
\begin{eqnarray*}
  R(t) &:=& r_0+\int_0^tp(s)ds-Z(t) \\
   &=& r_0+\int_0^tp(s)ds-\int_0^t\int_0^\infty zN^0(ds,dz)\,.
\end{eqnarray*}

Let $\pi(t)$ be the amount of the money invested in the risky asset at time $t$. We denote the surplus process by $X(t)$,
then we formulate the surplus process with delay, which is caused by the capital
inflow/outflow function from the insurer's current wealth. We suppose that the
capital inflow/outflow function is given by
$$
\varphi(t,X(t),\bar{Y}(t),U(t))=(\vartheta(t)+\xi)X(t)-\vartheta(t)\bar{Y}
(t)-\xi U(t)\,,
$$
where $\vartheta(t)\geq0$ is uniformly bounded function of $t$, $\xi\geq0$ is a
constant and
$$
Y(t)=\int_{t-\varrho}^te^{\zeta (s-t)}X(s)dW_1(s)\,; \ \ \ \ \
\bar{Y}(t)=\frac{Y(t)}{\int_{t-\varrho}^te^{\zeta (s-t)}ds}\,; \ \ \
U(t)=X(t-\varrho)\,.
$$

Here, $Y,\,\bar{Y},\,U$ represents respectively the  integrated, average and
pointwise delayed information of the wealth process in the interval
$[t-\varrho,t]$. $\zeta\geq0$ is the average parameter and $\varrho\geq0$ the
delay parameter. $W_1$ is an independent Brownian motion.

The parameters $\vartheta$ and $\xi$ represent the weights proportional to the past performance of $X-\bar{Y}$ and $X-U$, respectively. A good performance $(\varphi>0)$, may bring to the insurer more wealth, so that he can pay part of the wealth to the policyholders. Otherwise, a bad performance $(\varphi<0)$ may forces the insurer to use the reserve or look for further capital in the market to cover the losses in order to achieve the final performance. \\

\noindent {\bf Remark.} According to the definition of our capital
inflow/outflow function, we take a noisy memory into account, thus generalizing
the inflow/outflow function considered in Shen and Zeng \cite{Shen}. To the best
of our knowledge, this kind of noisy delay has just been applied in a stochastic
control problem recently by Dahl {\it et. al.} \cite{Dahl} using a maximum
principle techniques with Malliavin derivatives. Unlike in Dahl {\it et. al.} \cite{Dahl}, we suppose that the noisy delay is derived by an independent Brownian motion. We believe that this assumption is more realistic since the delay of the information may not be caused by the same source of randomness as the one driving the stock price. Furthermore, when the delay is driven by the same noisy with the asset price, the filtering theory we apply in the next section, fails to turn the model into one with complete observations, as the dynamics of $Y(t)$ would still be dependent on some hidden parameters. Under derivative pricing, such kind of delays have been applied to consider some stochastic volatility models, but with the delay driven by independent Poisson process, see, e.g., Swishchuk \cite{swishchuk}.  \\

Note that we can write the noisy memory information $Y$ in a differential form by
\begin{eqnarray}\label{integrated}
dY(t) &=& -\zeta Y(t)dt+X(t)dW_1(t)-e^{-\zeta\varrho}X(t-\varrho)dW_1(t-\varrho) \\ \nonumber
 &=& -\zeta Y(t)dt +X(t)(1-e^{-\zeta\varrho}\chi_{[0,T-\varrho]})dW_1(t) \ \ \
t\in[0,T]\,,
\end{eqnarray}
where $\chi_A$ denotes the characteristic function defined in a set $A$.

Then, the surplus process of the insurer is given by the following stochastic
delay differential equation (SDDE) with regime-switching

\begin{eqnarray}\label{surplus}
  &&dX(t)\\ \nonumber
&=& [p(t)+r(t)X(t)+\pi(t)(\alpha^{\Lambda}(t)-r(t))-\varphi(t,X(t),\bar{Y}(t),U(t))]
dt  \\\nonumber
   && +\pi(t)\beta(t)dW(t) +\pi(t)\int_{\mathbb{R}}zN(dt,dz)-\int_0^\infty
zN^0(dt,dz) \\ \nonumber
   &=&
\Bigl[p(t)+(r(t)-\vartheta(t)-\xi)X(t)+\pi(t)(\alpha^{\Lambda}(t)-r(t))+{
\vartheta}(t)\bar{Y}(t)  \\ \nonumber
   &&  +\xi U(t) +\sum_{j=1}^D\langle\Lambda(t-),e_j\rangle
\Bigl(\pi(t)\int_{\mathbb{R}}z\varepsilon_j(t)\nu_j(dz)
-\int_0^\infty\lambda_j(t)zf_j(dz)\Bigl)\Bigl]dt  \\ \nonumber
   && +\pi(t)\beta(t)dW(t) +\pi(t)\int_{\mathbb{R}}z\tilde{N}_{\Lambda}(dt,dz)
-\int_0^{\infty}z\tilde{N}^0_{\Lambda}(dt,dz)\,, \ \ \ t\in[0,T]\,, \\ \nonumber
   X(t) &=& x_0>0, \ \ \ t\in[-\varrho,0]\,.
\end{eqnarray}

The portfolio process $\pi(t)$ is said to be admissible if it satisfies the
following:
\begin{enumerate}
  \item $\pi(t)$ is $\mathcal{F}_t$-progressively measurable and
$\int_0^T|\pi(t)|^2dt<\infty$, $\mathbb{P}$-a.s.
  \item The SDDE \eqref{surplus} admits a unique strong solution;
  \item \begin{eqnarray*}
           &&
\sum_{j=1}^D\Bigl\{
\int_0^T|p(t)+(r_j(t)-\vartheta(t)-\xi)X(t)+\pi(t)(\alpha_j(t)-r_j(t))+{
\vartheta}(t)\bar{Y}(t) +\xi U(t)|dt \\
           && +\int_0^T\left[\pi^2(t)\beta^2(t)+
\int_{\mathbb{R}}(\pi(t))^2(t)z^2\varepsilon_j(t)\nu_j(dz)
+\int_0^{\infty}z^2\lambda_j(t)f_j(dz)\right]dt\Bigl\} \,\,<\, \infty\,~~\mathbb{P}-a.s.
        \end{eqnarray*}
\end{enumerate}
We denote the space of admissible investment strategy by $\mathcal{A}$.

\section{Reduction by the filtering theory}
As we are working with an unobservable Markov regime-switching model, one needs
to reduce the model into one with complete observations. We adopt the filtering
theory for this reduction. This is a classical approach and it has been widely
applied in stochastic control problems. See, for example, B\"auerle and Rieder
\cite{rieder}, Elliott {\it et. al.} \cite{elliott2008}, Elliott and Siu
\cite{elliott}, Siu \cite{siu2012}, and references therein. We proceed as in Siu
\cite{siu2014}.

Consider the following $\mathcal{F}_t$-adapted process
$\widehat{W}:=\{\widehat{W}(t),\, t\in[0,T]\}$ defined by
$$\widehat{W}(t):=W(t)+\int_0^t\frac{\alpha^{\Lambda}(s)-\hat{\alpha}^{\Lambda}
(s)}{\beta(s)}ds\,, \ \ \ t\in[0,T]\,,$$
where $\hat{\alpha}$ is the optional projection of $\alpha$ under $\mathbb{P}$, with respect to the filtration $\mathcal{F}_t$,
i.e.,
$\hat{\alpha}^{\Lambda}(t)=\mathbb{E}[\alpha^{\Lambda}(t)\mid\mathcal{F}_t]$, $\mathbb{P}$-a.s..
Then it was shown that $\widehat{W}$ is a Brownian motion. See e.g., Elliott and
Siu \cite{elliott} or Kallianpur \cite{kallianpur}, Lemma 11.3.1.

Let $\hat{\Lambda}$ be the optional projection of the Markov chain $\Lambda$.
For the jump part of the risk share $N$ and the insurance risk $N^0$, we
consider the following:
$$
\hat{\nu}(dt,dz):=\sum_{j=1}^D\langle\hat{\Lambda}(t-),
e_j\rangle\varepsilon_j(t)\nu_j(dz)dt\,\,\,\,
\rm{and}\,\,\,\,
\hat{\nu}^0(dt,dz):=\sum_{j=1}^D\langle\hat{\Lambda}(t-),
e_j\rangle\lambda_j(t)\nu_j^0(dz)dt\,.
$$
Define the compensated random measures $\widehat{N}(dt,dz)$ and
$\widehat{N}^0(dt,dz)$ by
\begin{eqnarray*}
  \widehat{N}(dt,dz) &:=& N(dt,dz)- \hat{\nu}(dt,dz) \\
  \widehat{N}^0(dt,dz) &:=& N^0(dt,dz)-\hat{\nu}^0(dt,dz)\,.
\end{eqnarray*}
Then, it can be shown that the following processes are martingales. (See Elliott
\cite{elliott1990}):
\begin{eqnarray*}
  \widehat{M} &:=& \int_0^t\int_{\mathbb{R}}z\widehat{N}(dt,dz) \\
  \widehat{M}^0 &:=&  \int_0^t\int_0^{\infty}z\widehat{N}^0(dt,dz)\,.
\end{eqnarray*}
Therefore, the surplus process $X(t)$ can be
written, under $\mathbb{P}$, as:
\begin{eqnarray}\label{surplushat}
&&dX(t)\\ \nonumber
&=&
\Bigl[p(t)+(r(t)-\vartheta(t)-\xi)X(t)+\pi(t)(\hat{\alpha}^{\Lambda}
(t)-r(t))+{\vartheta}(t)\bar{Y}(t)  \\ \nonumber
   &&  -\xi U(t) +\sum_{j=1}^D\langle\hat{\Lambda}(t-),e_j\rangle
\Bigl(\pi(t)\int_{\mathbb{R}}z\varepsilon_j(t)\nu_j(dz)
-\int_0^\infty\lambda_j(t)zf_j(dz)\Bigl)\Bigl]dt  \\ \nonumber
   && +\pi(t)\beta(t)d\widehat{W}(t)
+\pi(t)\int_{\mathbb{R}}z\widehat{N}_{\Lambda}(dt,dz)
-\int_0^{\infty}z\widehat{N}^0_{\Lambda}(dt,dz)\,, \ \ \ t\in[0,T]\,, \\
\nonumber
   X(t) &=& x_0>0, \ \ \ t\in[-\varrho,0]\,.
\end{eqnarray}

We then use the reference probability approach to derive a filtered estimate
$\hat{\Lambda}$ of the Markov chain $\Lambda$ following the discussions in Siu
\cite{siu2014}.

Let $\varphi(t)\in\mathbb{R}^D$, such that
$\varphi_j(t)=\alpha_j(t)-\frac{1}{2}\beta^2(t)$, $j=1,2,\ldots,D$. Define, for
any $t\in[0,T]$, the following functions
\begin{eqnarray*}
  \Psi_1(t) &:=& \int_0^t\langle\varphi(s),\Lambda(s)\rangle
ds+\int_0^t\beta(s)dW(s); \\
  \Psi_2(t) &:=& \int_0^t\int_{\mathbb{R}}zN(ds,dz); \\
  \Psi_3(t) &:=&  \int_0^t\int_0^{\infty}zN^0(ds,dz).
\end{eqnarray*}
Write $\mathbb{P}^*$, for a probability measure on $(\Omega,\mathcal{F})$, on
which the observation process does not depend on the Markov chain $\Lambda$. Define, for each
$j=1,2,\ldots,D$,
$$
F_j(t,z):=\frac{\lambda_j(t)f_j(dz)}{f(dz)} \ \ \ \ \rm{and} \ \ \ \
\mathcal{E}_j(t,z):=\frac{\varepsilon_j(t)\nu_j(dz)}{\nu(dz)}\,.
$$
Consider the following $\mathcal{F}_t$-adapted processes $\Gamma_1,\,\Gamma_2$
and $\Gamma_3$ defined by putting
\begin{eqnarray*}
  \Gamma_1(t) &:=&
\exp\Bigl(\int_0^t\beta^{-2}(s)\langle\varphi(s),\Lambda(s)\rangle d\Psi_1(s)
-\frac{1}{2}\int_0^t\beta^{-4}(s)\langle\varphi(s),\Lambda(s)\rangle^2ds\Bigl)\,
; \\
  \Gamma_2(t) &:=&
\exp\Bigl[-\int_0^t\sum_{j=1}^D\langle\Lambda(s-),e_j\rangle\left(\int_{\mathbb{
R}} (\mathcal{E}_j(s,z)-1)\nu(dz)\right)ds  \\
   && +\int_0^t\int_{\mathbb{R}} \left(\sum_{j=1}^D\langle\Lambda(s-),e_j\rangle
\ln(\mathcal{E}_j(s,z))\right)N(ds,dz)\Bigl]\,; \\
   \Gamma_3(t) &:=&
\exp\Bigl[-\int_0^t\sum_{j=1}^D\langle\Lambda(s-),e_j\rangle\left(\int_0^{\infty
} (F_j(s,z)-1)f(dz)\right)ds  \\
   && +\int_0^t\int_0^{\infty} \left(\sum_{j=1}^D\langle\Lambda(s-),e_j\rangle
\ln(F_j(s,z))\right)N^0(ds,dz)\Bigl]\,.
\end{eqnarray*}
Consider the $\mathcal{F}_t$-adapted process $\Gamma:=\{\Gamma(t),\,t\in[0,T]\}$
defined by
$$
\Gamma(t):=\Gamma_1(t)\cdot\Gamma_2(t)\cdot\Gamma_3(t).
$$
Note that the process $\Gamma$ is a local martingale and
$\mathbb{E}[\Gamma(T)]=1$. Under some strong assumptions, It can be shown that
$\Gamma$ is a true martingale. See, for instance, Proposition 2.5.1 in Delong
\cite{Delong}.

The main goal of the filtering process is to evaluate the
$\mathcal{F}_t$-optional projection of the Markov chain $\Lambda$ under
$\mathbb{P}$. To that end, let, for each $t\in[0,T]$,
$$
{\bf q}(t):=\mathbb{E}^*[\Gamma(t)\Lambda(t)\mid\mathcal{F}_t]\,,
$$
where $\mathbb{E}^*$ is an expectation under the reference probability measure
$\mathbb{P}^*$. The process ${\bf q}(t)$ is called an unnormalized filter of
$\Lambda(t)$.

Define, for each $j=1,2,\ldots,D$ the scalar valued process
$\gamma_j:=\{\gamma_j(t),\,t\in[0,T]$ by
\begin{eqnarray*}
  \gamma_j(t) &:=&
\exp\Bigl(\int_0^t\varphi_j(s)\beta^{-2}(s)d\Psi_1(s)-\frac{1}{2}
\int_0^t\varphi_j^2(s)\beta^{-4}(s)ds +\int_0^t(1-\varepsilon_j(s))ds \\
   && +\int_0^t(1-f_j(s))ds +\int_0^t\ln(\mathcal{E}_j(s))dN(s)
+\int_0^t\ln(F_j(s))dN^0(s)\Bigl)\,.
\end{eqnarray*}
Consider a diagonal matrix $\mathbf{L}(t):=\rm{{\bf
diag}}(\gamma_1(t),\gamma_2(t),\ldots,\gamma_D(t))$, for each $t\in[0,T]$.
Define the transformed unnormalized filter $\{\bar{\mathbf{q}}(t),\,t\in[0,T]\}$
by
$$
\bar{\mathbf{q}}(t):=\mathbf{L}^{-1}(t)\mathbf{q}(t)\,.
$$
Note that the existence of the inverse $\mathbf{L}^{-1}(t)$ is guaranteed by the
definition of $\mathbf{L}(t)$ and the positivity of $\gamma_j(t)$,
$j=1,2,\ldots,D$.

Then, it has been shown (see Elliott and Siu \cite{elliott}), that the
transformed unnormalized filter $\bar{\mathbf{q}}$ satisfies the following
linear order differential equation
$$
\frac{d\bar{\mathbf{q}}(t)}{dt}:=\mathbf{L}^{-1}(t)A(t)\mathbf{L}(t)\bar{\mathbf
{q}}(t)\,, \ \ \ \ \ \bar{\mathbf{q}}(0)=\mathbf{q}(0)=\mathbb{E}[\Lambda(0)]\,.
$$
Hence, by a version of the Bayes rule, the optimal estimate $\hat{\Lambda}(t)$
of the Markov chain $\Lambda(t)$ is given by
$$
\hat{\Lambda}:=\mathbb{E}[\Lambda(t)\mid\mathcal{F}_t]=\frac{\mathbb{E}^*[
\Gamma(t)\Lambda(t)
\mid\mathcal{F}_t]}{\mathbb{E}^*[\Gamma(t)\mid\mathcal{F}_t]}=\frac{\mathbf{q}
(t)}{\langle\mathbf{q}(t),\mathbf{1}\rangle}\,.
$$

\section{Risk-based optimal investment problem}

In this section, we introduce the optimal investment problem of an insurer with
regime-switching and delay. We consider a problem where the objective is to
minimize the risk described by the convex risk measure, with the insurer not only
concerned with the terminal wealth, but also with the integrated noisy memory
surplus over the period $[T-\varrho,T]$. This \textbf{problem} is then described
as follows:
 Find the investment strategy $\pi(t)\in\mathcal{A}$ which minimizes the risks
of the terminal surplus and the integrated surplus, i.e., $X(T)+\kappa Y(T)$,
where $\kappa\geq0$ denotes the weight between $X(T)$ and $Y(T)$. This allows us to incorporate the terminal wealth as well as the delayed wealth at the terminal time $T$ in the performance functional.

Since we are dealing with a measure of risk, we will use the concept of convex
risk measures introduced in F\"ollmer and Schied \cite{follmer} and Frittelli
and Rosazza \cite{frittelli}. Which is the generalization of the concept of
coherent risk measures proposed by Artzner {\it et. al.} \cite{artzner}.

\begin{definition}
Let $\mathcal{S}$ be a space of all lower bounded $\{\mathcal{G}_t\}_{t\in[0,T]}$-measurable
random variables. A convex risk measure on $\mathcal{S}$ is a map
$\rho:\mathcal{S}\rightarrow\mathbb{R}$ such that:
\begin{enumerate}
  \item  ({\it translation}) If $\epsilon\in\mathbb{R}$ and $\chi\in\mathcal{S}$, then $\rho(\chi+\epsilon)=\rho(\chi)-\epsilon$;
  \item  ({\it monotonicity}) For any $\chi_1,\chi_2\in\mathcal{S}$, if $\chi_1(\omega)\leq \chi_2(\omega)$; $\omega\in\Omega$, then $\rho(\chi_1)\geq\rho(\chi_2)$;
  \item  ({\it convexity}) For any $\chi_1,\chi_2\in\mathcal{S}$ and $\varsigma\in(0,1)$,
  $$
  \rho(\varsigma \chi_1+(1-\varsigma)\chi_2)\leq\varsigma\rho(\chi_1)+(1-\varsigma)\rho(\chi_2)\,.
  $$
\end{enumerate}

\end{definition}

Following the general representation of the convex risk measures (see e.g.,
Theorem 3, Frittelli and Rasozza \cite{frittelli}), also applied by Mataramvura
and \O ksendal \cite{Mataramvura}, Elliott and Siu \cite{Elliott2011}, Meng and
Siu \cite{Meng}, among others, we assume that the risk measure $\rho$ under
consideration in this paper, is as follows:
$$
\rho(\chi)=\sup_{Q\in\mathcal{M}_a}\{\mathbb{E}^{\mathbb{Q}}[-\chi]-\eta(Q)\},
$$
where $\mathbb{E}^{\mathbb{Q}}$ is the expectation under $\mathbb{Q}$, for the
family $\mathcal{M}_a$ of probability measures and for some penalty function
$\eta:\mathcal{M}_a\rightarrow\mathbb{R}$.

In order to specify the penalty function, we first describe a family
$\mathcal{M}_a$ of all measures $Q$ of Girsanov type. We consider a robust
modeling setup, given by a probability measure $\mathbb{Q}:=Q^{\theta_0,
\theta_1,\theta_2}$, with the Radon-Nikodym derivative given by
$$
\frac{d\mathbb{Q}}{d\mathbb{P}}\Bigl|_{\mathcal{F}_t}=G^{\theta_0,\theta_1,\theta_2}(t)\,
, \ \ \ 0\leq t\leq T\,.
$$

The Radon-Nikodym $G^{\theta_0,\theta_1,\theta_2}(t)\,, \ \ \ t\in[0,T+\varrho],$ is
given by

\begin{eqnarray} \label{martingalemeasure}
  dG^{\theta_0,\theta_1,\theta_2}(t) &=& G^{\theta_0,\theta_1,\theta_2}(t^-)\Bigl[\theta_0(t)d\widehat{W}(t) +\theta_1(t)dW_1(t) +\int_0^{\infty}\theta_0(t)\widehat{N}^0_{\Lambda}(dt,dz) \\ \nonumber
  &&  +\int_{\mathbb{R}}\theta_2(t,z)\widehat{N}_{\Lambda}(dt,dz)\Bigl]\,, \\ \nonumber
  G^{\theta_0,\theta_1,\theta_2}(0) &=& 1, \\ \nonumber
  G^{\theta_0,\theta_1,\theta_2}(t) &=& 0\,, \ \ \ t\in[-\varrho,0)\,.
\end{eqnarray}
The set $\Theta:=\{\theta_0,\theta_1,\theta_2\}$ is considered as a set of scenario
control. We say that $\Theta$ is admissible if $\theta_2(t,z)>-1$ and
$$
\mathbb{E}\left[\int_0^T\Bigl\{\theta_0^2(t)+\theta_1^2(t)+\int_{\mathbb{R}}\theta_2^2(t,
z)\nu_{\Lambda}(dz)\Bigl\}dt\right]<\infty\,.
$$
Then, the family $\mathcal{M}_a$ of probability measures  is given by
$$
\mathcal{M}_a:=\mathcal{M}(\Theta)=\{\mathbb{Q}^{\theta_0,\theta_1,\theta_2}\,:
\,(\theta_0,\theta_1,\theta_2)\in\Theta\}\,.
$$

Let us now specify the penalty function $\eta$. Suppose that for each $(\pi,
\theta_0,\theta_1,\theta_2)\in\mathcal{A}\times\Theta$ and $t\in[0,T]$,
$\pi(t)\in\mathbf{U}_1$ and $\theta(t)=(\theta_0(t),\theta_1(t),\theta_2(t,\cdot))\in\mathbf{U_2}$, where
$\mathbf{U}_1$ and $\mathbf{U}_2$ are compact metric spaces in $\mathbb{R}$ and
$\mathbb{R}^3$.

Let
$\ell:[0,T]\times\mathbb{R}\times\mathbb{R}\times\mathbb{R}\times\mathbf{U}
_1\times\mathbf{U}_2 \rightarrow\mathbb{R}$ and
$h:\mathbb{R}\times\mathbb{R}\rightarrow\mathbb{R}$ be two bounded measurable
convex functions in $\theta(t)\in\mathbf{U}_2$ and
$(X(T),Y(T))\in\mathbb{R}\times\mathbb{R}$, respectively. Then, for each
$(\pi,\theta)\in\mathcal{A}\times\Theta$,
$$
\mathbb{E}\left[\int_0^T|\ell(t,X(t),Y(t),Z(t),\pi(t),\theta_0(t),
\theta_1(t),\theta_2(t,\cdot))|dt+ |h(X(T),Y(T))|\right]<\infty\,.
$$

As in Mataramvura and \O ksendal \cite{Mataramvura}, we consider, for each
$(\pi,\theta)\in\mathcal{A}\times\Theta$, a penalty function $\eta$ of the form

$$
\eta(\pi,\theta_0,\theta_1,\theta_2):=\mathbb{E}\left[\int_0^T\ell(t,X(t),Y(t),Z(t),
\pi(t),\theta_0(t),\theta_1(t),\theta_2(t,\cdot))dt+ h(X(T),Y(T))\right]\,.
$$

Then, we define a convex risk measure for the terminal wealth and the integrated
wealth of an insurer, i.e., $X(T)+\kappa Y(T)$, for $\kappa\geq0$, given the
information $\mathcal{F}_t$ associated with the family of probability measures
$\mathcal{M}_a$  and the penalty function $\eta$, as follows:

$$
\rho(X(T),Y(T)):=\sup_{(\theta_0,\theta_1,\theta_2)\in\Theta}\left\{\mathbb{E}^{\mathbb{Q
}}[-(X^{\pi}(T)+\kappa Y^{\pi}(T))]-\eta(\pi,\theta_0,\theta_1,\theta_2)\right\}\,.
$$

As in Elliott and Siu \cite{Elliott2011}, the main objective of  the insurer is
to select the optimal investment process $\pi(t)\in\mathcal{A}$ so as to
minimizes the risks described by $\rho(X(T),Y(T))$. That is, the optimal problem
of an insurer is:
\begin{equation}\label{objectivefunctional}
\mathcal{J}(x):=\inf_{\pi\in\mathcal{A}}\left\{\sup_{(\theta_0,
\theta_1,\theta_2)\in\Theta} \left\{\mathbb{E}^{\mathbb{Q}}[-(X^{\pi}(T)+\kappa
Y^{\pi}(T))]-\eta(\pi,\theta_0,\theta_1,\theta_2)\right\}\right\}\,.
\end{equation}

Note that $\mathbb{E}^{\mathbb{Q}}[-(X^{\pi}(T)+\kappa
Y^{\pi}(T))]=\mathbb{E}[-(X^{\pi}(T)+\kappa
Y^{\pi}(T))G^{\theta_0,\theta_1}(T)]$ (See Cuoco \cite{cuoco} or Karatzas and
Shreve \cite{karatzas} for more details). Then from the form of the penalty
function,

\begin{eqnarray*}
\bar{\mathcal{J}}(x) &=& \inf_{\pi\in\mathcal{A}}\sup_{(\theta_0,\theta_1,\theta_2)\in\Theta}
\mathbb{E}\Bigl[-(X^{\pi}(T)+\kappa Y^{\pi}(T))G^{\theta_0,\theta_1,\theta_2}(T) \\
&& -\int_0^T\ell(t,X(t),Y(t),Z(t),\pi(t),\theta_0(t),\theta_1(t),\theta_2(t,\cdot))dt-
h(X(T),Y(T))\Bigl] \\
&=& \mathcal{J}(x), \ \ \ \rm{say}.
\end{eqnarray*}

For each $(\pi,\theta)\in\mathcal{A}\times\Theta$, suppose that

\begin{eqnarray*}
\mathcal{V}^{\pi,\theta}(x) &:=& \mathbb{E}\Bigl[-(X^{\pi}(T)+\kappa
Y^{\pi}(T))G^{\theta_0,\theta_1,\theta_2}(T) \\
&& -\int_0^T\ell(t,X(t),Y(t),Z(t),\pi(t),\theta_0(t),\theta_1(t),\theta_2(t,\cdot))dt-
h(X(T),Y(T))\Bigl]\,.
\end{eqnarray*}
Then,
$$
\mathcal{J}(x)= \inf_{\pi\in\mathcal{A}}\sup_{(\theta_0,\theta_1,\theta_2)\in\Theta}
\mathcal{V}^{\pi,\theta}(x)=\mathcal{V}^{\pi^*,\theta^*}(x)\,,
$$
that is, the insurer selects an optimal investment strategy $\pi$ so as to
minimize the maximal risks, whilst the market reacts by selecting a probability
measure indexed by $((\theta_0,\theta_1,\theta_2))\in\Theta$ corresponding to the
worst-case scenario, where the risk is maximized. To solve this game problem,
one must select the optimal strategy $(\pi^*,\theta_0^*,\theta_1^*,\theta_2^*)$ from the
insurer and the market, respectively, as well as the optimal value function
$\mathcal{J}(x)$.

\section{The BSDE approach to a game problem}

In this section, we solve the risk-based optimal investment problem of an insurer using delayed BSDE with jumps. Delayed BSDEs may arise in insurance and finance, when one wants to find an investment strategy which should replicate a liability or meet a purpose depending on the past values of the portfolio. For instance, under participating contracts in life insurance endowment contracts, we have a so called {\it performance-linked payoff}, that is, the payoff from the policy is related to the performance of the portfolio held by the insurer. Thus, the current portfolio and the past values of the portfolio have an impact on the final value of the liability. For more discussions on this and more applications of delayed BSDEs see Delong \cite{delongl2010}.

We first consider the following notation in order to establish the existence and uniqueness result of a delayed BSDE with jumps.
\begin{itemize}
\item $\mathbb{L}^2_{-\varrho}(\mathbb{R})$- the space of measurable functions
$k:[-\varrho,0]\mapsto\mathbb{R}$, such that $\int_{-\varrho}^0|k(t)|^2dt<\infty$;
\item $\mathbb{S}^2_{-\varrho}(\mathbb{R})$- the space of bounded measurable functions
$y:[\varrho,0]\mapsto\mathbb{R}$ such that
$$\sup|y(t)|^2<\infty\,;$$
\item $\mathbb{H}_{-\varrho,\nu}^2$- the space of product measurable functions
$\upsilon:[-\varrho,0]\times\mathbb{R}$, such that
\begin{equation*}
\int_0^T\int_{\mathbb{R}}|\upsilon(t,z)|^2\nu(dz)dt
<\infty\,;
\end{equation*}
  \item $\mathbb{L}^2(\mathbb{R})$- the space of random variables
$\xi:\Omega\mapsto\mathbb{R}$, such that $\mathbb{E}[\,|\xi|^2]<\infty$;
  \item $\mathbb{H}^2(\mathbb{R})$- the space of measurable functions
$K:\mathbb{R}\mapsto\mathbb{R}$ such that
$$\mathbb{E}\left[\int_{\mathbb{R}}|K(t)|^2dt\right]<\infty\,;$$
  \item $\mathbb{S}^2(\mathbb{R})$- the space of adapted c\`adl\`ag processes
$Y:\Omega\times[0,T]\mapsto\mathbb{R}$ such that
$$\mathbb{E}[\sup|Y(t)|^2]<\infty$$ and
  \item $\mathbb{H}_\nu^2$- the space of predictable processes
$\Upsilon:\Omega\times[0,T]\times\mathbb{R}\mapsto\mathbb{R}$, such that
\begin{equation*}
\mathbb{E}\left[\int_0^T\int_{\mathbb{R}}|\Upsilon(t,z)|^2\nu(dz)dt\right]
<\infty.
\end{equation*}

\end{itemize}

Define the following delayed BSDE with jumps:
\begin{eqnarray}\label{bsde}
  d\mathcal{Y}(t) &=& -\mathcal{W}(t,\pi(t),\theta(t))dt+K_1(t)d\widehat{W}(t)+K_2(t)dW_1(t) \\ \nonumber
   && +\int_{\mathbb{R}}\Upsilon_1(t,z)\widehat{N}_{\Lambda}(dt,dz)
+\int_0^{\infty}\Upsilon_2(t,z)\widehat{N}_{\Lambda}^0(dt,dz) ; \\ \nonumber
  \mathcal{Y}(T) &=& h(X(T),Y(T))\,,
\end{eqnarray}
where
\begin{eqnarray*}
\mathcal{W}(t,\pi(t),\theta(t)) &:=& \mathcal{G}(t, \mathcal{Y}(t),\mathcal{Y}(t-\varrho), K_1(t), K_1(t-\varrho),K_2(t), K_2(t-\varrho),
\Upsilon_1(t,\cdot), \\
&&  \ \ \ \ \ \Upsilon_1(t-\varrho,\cdot),\Upsilon_2(t,\cdot), \Upsilon_2(t-\varrho,\cdot),\pi(t),\theta(t))\,.
\end{eqnarray*}
We assume that the generator $\mathcal{W}:\Omega\times[0,T]\times\mathbb{S}^2(\mathbb{R})\times\mathbb{S}^2_{-\varrho}(\mathbb{R}) \times\mathbb{H}^2(\mathbb{R})\times\mathbb{H}^2_{-\varrho}(\mathbb{R})\times\mathbb{H}_\nu^2(\mathbb{R}) \times\mathbb{H}_{-\varrho,\nu}^2(\mathbb{R})\mapsto\mathbb{R}$ satisfy the following Lipschtz continuous condition, i.e.,
there exists a constant $C>0$ and a probability measure $\eta$ on $([-\varrho,0],\mathcal{B}([-\varrho,0]))$ such that
\begin{eqnarray*}
  && \mathcal{W}(t,\pi(t),\theta(t))- \tilde{\mathcal{W}}(t,\pi(t),\theta(t)) \\
  &\leq& C\Bigl(\int_{-\varrho}^0|y(t+\zeta) -\tilde{y}(t+\zeta)|^2\eta(d\zeta) +\int_{-\varrho}^0|k_1(t+\zeta) -\tilde{k}_1(t+\zeta)|^2\eta(d\zeta)  \\
   && +\int_{-\varrho}^0|k_2(t+\zeta) -\tilde{k}_2(t+\zeta)|^2\eta(d\zeta) +\int_{-\varrho}^0\int_{\mathbb{R}}|\upsilon_1(t+\zeta,z) -\tilde{\upsilon}_1(t+\zeta,z)|^2\nu(dz)\eta(d\zeta)  \\
   && +\int_{-\varrho}^0\int_{\mathbb{R}}|\upsilon_2(t+\zeta,z) -\tilde{\upsilon}_2(t+\zeta,z)|^2\nu(dz)\eta(d\zeta) +\int_{0}^T|y(t) -\tilde{y}(t)|^2dt  \\
   && +\int_{0}^T|k_2(t) -\tilde{k}_2(t)|^2dt +\int_{0}^T\int_{\mathbb{R}}|\upsilon_1(t,z) -\tilde{\upsilon}_1(t,z)|^2\nu(dz)dt \\
   && +\int_{0}^T\int_{\mathbb{R}}|\upsilon_2(t,z) -\tilde{\upsilon}_2(t,z)|^2\nu(dz)dt\Bigl)\,.
\end{eqnarray*}
Then, if $h\in\mathbb{L}^2$ and the above Lipschitz condition is satisfied, one can prove  the existence and uniqueness solution $(\mathcal{Y},K_1,K_2,\Upsilon_1,\Upsilon_2)\in\mathbb{S}^2(\mathbb{R})\times\mathbb{H}^2(\mathbb{R})\times\mathbb{H}^2(\mathbb{R}) \times\mathbb{H}^2_{\nu}(\mathbb{R})\times\mathbb{H}^2_{\nu}(\mathbb{R})$ of a delayed BSDE with jumps \eqref{bsde}. See Delong and Imkeller \cite{delong2010} and Delong \cite{Delong} for more details. In practice, $\mathcal{Y}$ denotes a replicating portfolio, $K_1,K_2,\Upsilon_1,\Upsilon_2$ represent the replicating strategy, $h(X(T),Y(T))$ is a terminal liability and $\mathcal{G}$ models the stream liability during the contract life-time.

The key result for solving our delayed stochastic differential game problem is based on the following theorem.

\begin{theorem}\label{theorembsde}
Suppose that there exists a strategy $(\hat{\pi}(t),
\hat{\theta}(t))\in\mathbf{U}_1\times\mathbf{U}_2$ such that
\begin{eqnarray}\label{esssuph}
\mathcal{W}(t, y, k_1,k_2, \upsilon(\cdot), \hat{\pi}(t), \hat{\theta}(t)) &=&
\inf_{\pi\in\mathcal{A}}\sup_{(\theta_0,\theta_1,\theta_2)\in\Theta} \mathcal{G}(t, y, k_1,k_2,
\upsilon(\cdot), \pi, \theta) \\ \nonumber
&=& \sup_{(\theta_0,\theta_1,\theta_2)\in\Theta}\inf_{\pi\in\mathcal{A}} \mathcal{G}(t,
y, k_1,k_2, \upsilon(\cdot), \pi, \theta)\,,
\end{eqnarray}
that is, $\mathcal{W}$ satisfy the Isaac's condition.
Furthermore, suppose that there exists a unique solution
$(\mathcal{Y}^{\pi,\theta}(t), K_1^{\pi,\theta}(t),K_2^{\pi,\theta}(t),
\Upsilon_1^{\pi,\theta}(t,\cdot),\Upsilon_2^{\pi,\theta}(t,\cdot))\in\mathbb{S}^2(\mathbb{R})\times\mathbb{H}^2(\mathbb{R}) \times\mathbb{H}^2(\mathbb{R}) \times\mathbb{H}^2_{\nu}(\mathbb{R})\times\mathbb{H}^2_{\nu}(\mathbb{R})$ of the BSDE
\eqref{bsde}, for all $(\pi,\theta)\in\mathcal{A}\times\Theta$.

Then, the value function $\mathcal{J}(x)$ is given by
$\mathcal{Y}^{\hat{\pi},\hat{\theta}}(t)$. Moreover, the optimal strategy of the
problem \eqref{objectivefunctional} is given by

\begin{equation}\label{optimalgeneral}
\left\{
  \begin{array}{ll}
  \pi^*(t)=\hat{\pi}(t,Y(t),K_1(t),K_2(t),\Upsilon_1(t,\cdot),\Upsilon_2(t,\cdot))\,, &
\hbox{} \\

\theta^*(t)=\hat{\theta}(t,Y(t),K_1(t),K_2(t),\Upsilon_1(t,\cdot),\Upsilon_2(t,\cdot))\,.
& \hbox{}
  \end{array}
\right.
\end{equation}

\end{theorem}

\begin{proof}
The proof is based on the comparison principle for BSDEs with jumps as follows,
(see Theorem 3.2.1 in Delong \cite{Delong}).
Define three generators $\phi_1, \ \phi_2$ and $\phi_3$ by
\begin{eqnarray*}
\phi_1(t, y,k_1,k_2, \upsilon(\cdot)) &=& \mathcal{W}(t, y, k_1,k_2,
\upsilon_1(\cdot),\upsilon_2(\cdot), \hat{\pi}(t), \theta(t)) \\
  \phi_2(t, y,k_1,k_2, \upsilon(\cdot))&=& \mathcal{W}(t, y, k_1,k_2,
\upsilon_1(\cdot),\upsilon_2(\cdot), \hat{\pi}(t), \hat{\theta}(t)) \\
  \phi_3(t, y, k_1,k_2, \upsilon(\cdot)) &=& \mathcal{W}(t, y, k_1,k_2,
\upsilon_1(\cdot),\upsilon_2(\cdot), \pi(t), \hat{\theta}(t))
\end{eqnarray*}
and the corresponding BSDEs

\begin{eqnarray}\nonumber
  d\mathcal{Y}_1(t) &=& -\phi_1(t, y,k_1,k_2, \upsilon(\cdot))dt+K_1(t)d\widehat{W}(t) +K_2(t)dW_1(t)
+\int_{\mathbb{R}}\Upsilon_1(t,z)\widehat{N}_{\Lambda}(dt,dz) \\ \nonumber
  && +\int_0^{\infty}\Upsilon_2(t,z)\widehat{N}_{\Lambda}^0(dt,dz)\,, \\
\nonumber
  \mathcal{Y}_1(T) &=& h(X(T),Y(T))\,.
\end{eqnarray}

\begin{eqnarray}\nonumber
  d\mathcal{Y}_2(t) &=& -\phi_2(t, y,k_1,k_2, \upsilon(\cdot))dt+K(t)d\widehat{W}(t) +K_2(t)dW_1(t)
+\int_{\mathbb{R}}\Upsilon_1(t,z)\widehat{N}_{\Lambda}(dt,dz) \\ \nonumber
  && +\int_0^{\infty}\Upsilon_2(t,z)\widehat{N}_{\Lambda}^0(dt,dz) \\ \nonumber
  \mathcal{Y}_2(T) &=& h(X(T),Y(T))\,.
\end{eqnarray}
and
\begin{eqnarray}\nonumber
  d\mathcal{Y}_3(t) &=& -\phi_3(t, y,k_1,k_2, \upsilon(\cdot))dt+K(t)d\widehat{W}(t) +K_2(t)dW_1(t)
+\int_{\mathbb{R}}\Upsilon_1(t,z)\widehat{N}_{\Lambda}(dt,dz) \\ \nonumber
  && +\int_0^{\infty}\Upsilon_2(t,z)\widehat{N}_{\Lambda}^0(dt,dz)\,, \\
\nonumber
  \mathcal{Y}_3(T) &=& h(X(T),Y(T))\,.
\end{eqnarray}

From \eqref{esssuph}, we have
$$
\phi_1(t, y, k_1,k_2, \upsilon(\cdot))\leq \phi_2(t, y, k_1,k_2, \upsilon(\cdot))\leq
\phi_3(t, y, k_1,k_2, \upsilon(\cdot)).
$$
Then, by comparison principle, $\mathcal{Y}_1(t)\leq
\mathcal{Y}_2(t)=\mathcal{J}(x)\leq \mathcal{Y}_3(t)$, for all $t\in[0,T]$. By
uniqueness, we get $\mathcal{Y}_2(t)=\mathcal{V}^{\pi^*,\theta^*}$. Hence, the
optimal strategy is given by \eqref{optimalgeneral}.
\end{proof}

In order to solve our main problem, note that from the dynamics of the processes
$X(t),\,Y(t)$ and $G^{\theta_0,\theta_1,\theta_2}$ in \eqref{surplus}, \eqref{integrated}
and \eqref{martingalemeasure}, respectively and applying the It\^o's differentiation rule for delayed SDEs with jumps (See Ba\~nos {\it et. al.} \cite{banos}, Theorem 3.8), we have
formula, we have:
\begin{eqnarray*}
 && d[(X(t)+\kappa Y(t))G^{\theta_0,\theta_1,\theta_2}(t)] \\
  &=&
G^{\theta_0,\theta_1,\theta_2}(t)\Bigl[p(t)+(r(t)-\vartheta(t)-\xi)X(t)+\pi(t)(\hat{
\alpha}^{\Lambda}(t)-r(t))+({\vartheta}(t)-\kappa\zeta)\bar{ Y}(t)+\xi U(t)  \\
   && +\pi(t)\beta(t)\theta_0(t) +\theta_1(t)\kappa
X(t)(1-e^{-\zeta\varrho}\chi_{[0,T-\varrho]}) \\
   && +\sum_{j=1}^D\langle\hat{\Lambda}(t-),e_j\rangle
\Bigl(\pi(t)\int_{\mathbb{R}}z\theta_2(t,z)\varepsilon_j(t)\nu_j(dz)
-\int_0^\infty\lambda_j(t)z(1+\theta_0(t))f_j(dz)\Bigl) \Bigl]dt  \\
    && +
G^{\theta_0,\theta_1,\theta_2}(t)[(\pi(t)\beta(t)+X(t)\theta_0(t)) d\widehat{W}(t) +(\theta_1(t)+X(t)(1-e^{-\zeta\varrho}\chi_{[0,T-\varrho]}))dW_1(t)] \\
    && +G^{\theta_0,\theta_1,\theta_2}(t)\int_{\mathbb{R}}[(1+\theta_2(t,z))\pi(t)z
+X(t)\theta_2(t,z)] \widehat{N}_{\Lambda}(dt,dz) \\
    && -G^{\theta_0,\theta_1,\theta_2}(t)
\int_0^{\infty}[(1+\theta_0(t))z-X(t)\theta_0(t)]\widehat{N}^0_{\Lambda}(dt,dz)\
,.
\end{eqnarray*}
Thus, for each $(\pi,\theta_0,\theta_1,\theta_2)$,
\begin{eqnarray*}
   && \mathcal{J}(x) \\
  &=& \mathbb{E}\Bigl\{-\int_0^T\Bigl[
G^{\theta_0,\theta_1,\theta_2}(t)\Bigl[p(t)+(r(t)-\vartheta(t)-\xi)X(t)+\pi(t)(\hat{
\alpha}^{\Lambda}(t)-r(t)) +({\vartheta}(t)-\kappa\zeta)\bar{Y}(t) \\
   && -\xi U(t) +\pi(t)\beta(t)\theta_0(t) +\theta_1(t)\kappa
X(t)(1-e^{-\zeta\varrho}\chi_{[0,T-\varrho]}) \\
   && +\sum_{j=1}^D\langle\hat{\Lambda}(t-),e_j\rangle
\Bigl(\pi(t)\int_{\mathbb{R}}z\theta_2(t,z)\varepsilon_j(t)\nu_j(dz)
-\int_0^\infty\lambda_j(t)z(1+\theta_0(t))f_j(dz)\Bigl)\Bigl]  \\
   && +\ell(t,X(t),Y(t),U(t),\pi(t),\theta_0(t),\theta_1(t),\theta_2(t,\cdot))\Bigl]dt -
h(X(T),Y(T))\Bigl\}\,.
\end{eqnarray*}
We now define, for each $(t,X,Y,U,\pi,\theta_0,\theta_1,\theta_2)\in
[0,T]\times\mathbb{R}\times\mathbb{R}\times\mathbb{R}\times\mathbf{U}
_1\times\mathbf{U}_2$, a function
\begin{eqnarray*}
  && \tilde{\ell}(t,X(t),Y(t),U(t),\pi(t),\theta_0(t),\theta_1(t),\theta_2(t,\cdot)) \\
  &=&
G^{\theta_0,\theta_1,\theta_2}(t)\Bigl[p(t)+(r(t)-\vartheta(t)-\xi)X(t)+\pi(t)(\hat{
\alpha}^{\Lambda}(t)-r(t)) +({\vartheta}(t)-\kappa\zeta)\bar{Y}(t)-\xi U(t)  \\
   && +\pi(t)\beta(t)\theta_0(t) +\theta_1(t)\kappa
X(t)(1-e^{-\zeta\varrho}\chi_{[0,T-\varrho]}) \\
   && +\sum_{j=1}^D\langle\hat{\Lambda}(t-),e_j\rangle
\Bigl(\pi(t)\int_{\mathbb{R}}z\theta_2(t,z)\varepsilon_j(t)\nu_j(dz)
-\int_0^\infty\lambda_j(t)z(1+\theta_0(t))f_j(dz)\Bigl)\Bigl] \\
   &&    +\ell(t,X(t),Y(t),U(t),\pi(t),\theta_0(t),\theta_1(t),\theta_2(t,\cdot))\,.
\end{eqnarray*}
Then,
\begin{eqnarray*}
&& \mathcal{J}(x) \\
&=& -x_0+\mathbb{E}\Bigl[-\int_0^T
\tilde{\ell}(t,X(t),Y(t),U(t),\pi(t),\theta_0(t),\theta_1(t),\theta_2(t,\cdot))dt-h(X(T),
Y(T))\Bigl]\,.
\end{eqnarray*}
Define, for each $(\pi,\theta)\in\mathcal{A}\times\Theta$, a functional
$$
\tilde{\mathcal{J}}(x)=\mathbb{E}\Bigl[-\int_0^T
\tilde{\ell}(t,X(t),Y(t),U(t),\pi(t),\theta_0(t),\theta_1(t),\theta_2(t,\cdot))dt-h(X(T),
Y(T))\Bigl]\,.
$$
Then, the stochastic differential delay game problem discussed in the previous
section is equivalent to the following problem:
$$
\tilde{\mathcal{V}}(t,x)=\inf_{\pi\in\mathcal{A}}\sup_{(\theta_0,
\theta_1,\theta_2)\in\Theta} \tilde{\mathcal{J}}(x)\,.
$$
We now define the Hamiltonian of the aforementioned game problem
$\mathcal{H}:[0,T]\times\mathbb{R}\times\mathbb{R}\times\mathbb{R}\times\mathbb{
R}
\times\mathbb{R}\times\mathbb{R}\times\mathbf{U}_1\times\mathbf{U}
_2\rightarrow\mathbb{R}$ as follows:
\begin{eqnarray*}
&&\mathcal{H}(t,X(t),Y(t),U(t), K_1(t),K_2(t),\Upsilon_1(t,\cdot),
\Upsilon_2(t,\cdot),\pi(t),\theta_0(t),\theta_1(t),\theta_2(t,\cdot)) \\
&:=& -\tilde{\ell}(t,X(t),Y(t),U(t),\pi(t),\theta_0(t),\theta_1(t),\theta_2(t,\cdot))\,.
\end{eqnarray*}

In order for the Hamiltonian $\mathcal{H}$ to satisfy the Issac's condition, we
require that $\mathcal{H}$ is convex in $\pi$ and
concave in $(\theta_0,\theta_1,\theta_2)$. Moreover, for the existence and uniqueness
solution of the corresponding delayed BSDE with jumps, the Hamiltonian should satisfy
the Lipschitz condition. From the boundedness of the associate
parameters, we prove that $\mathcal{H}$ is indeed Lipschitz.

\begin{lemma}\label{lemma}
The Hamiltonian $\mathcal{H}$ is Lipschitz continuous in $X$, $Y$ and $U$.
\end{lemma}

\begin{proof}
Since $(\pi,(\theta_0,\theta_1,\theta_2))\in\mathbf{U}_1\times \mathbf{U}_2$ and $\ell$ is bounded, $\tilde{\ell}$ is bounded. Then $\mathcal{H}$ is uniformly bounded with respect to $(t,X(t),Y(t),U(t),\pi(t),\theta_0(t),\theta_1(t),\theta_2(t,\cdot))$.  To prove the Lipschitz condition, we suppose that $\mathcal{H}$ is not Lipschitz continuous in $(K_1(t),K_2(t),\Upsilon_1(t,\cdot),
\Upsilon_2(t,\cdot))$, uniformly in $(t,X(t),Y(t),U(t))$. Then, there exist two points $(K_1(t),K_2(t),\Upsilon_1(t,\cdot),
\Upsilon_2(t,\cdot))$, $(\tilde{K}_1(t),\tilde{K}_2(t),\tilde{\Upsilon}_1(t,\cdot),\tilde{\Upsilon}_2(t,\cdot))$ such that
\begin{eqnarray*}
|\mathcal{H}(t,X(t),Y(t),U(t), K_1(t),K_2(t),\Upsilon_1(t,\cdot),
\Upsilon_2(t,\cdot),\pi(t),\theta_0(t),\theta_1(t),\theta_2(t,\cdot)) && \\
 -\mathcal{H}(t,X(t),Y(t),U(t),\tilde{K}_1(t),\tilde{K}_2(t),\tilde{\Upsilon}_1(t,\cdot),\tilde{\Upsilon}_2(t,\cdot) ,\pi(t),\theta_0(t),\theta_1(t),\theta_2(t,\cdot))| &&
 \end{eqnarray*} is unbounded. However, since $\mathcal{H}$ is uniformly bounded with respect to \\$(t,X(t), Y(t), U(t), \pi(t),\theta_0(t),\theta_1(t),\theta_2(t,\cdot))$, we have
 \begin{eqnarray*}
|\mathcal{H}(t,X(t),Y(t),U(t), K_1(t),K_2(t),\Upsilon_1(t,\cdot),
\Upsilon_2(t,\cdot),\pi(t),\theta_0(t),\theta_1(t),\theta_2(t,\cdot)) && \\
 -\mathcal{H}(t,X(t),Y(t),U(t),\tilde{K}_1(t),\tilde{K}_2(t),\tilde{\Upsilon}_1(t,\cdot),\tilde{\Upsilon}_2(t,\cdot) ,\pi(t),\theta_0(t),\theta_1(t),\theta_2(t,\cdot))| &\leq& \\
 |\mathcal{H}(t,X(t),Y(t),U(t), K_1(t),K_2(t),\Upsilon_1(t,\cdot),
\Upsilon_2(t,\cdot),\pi(t),\theta_0(t),\theta_1(t),\theta_2(t,\cdot))| && \\
+|\mathcal{H}(t,X(t),Y(t),U(t),\tilde{K}_1(t),\tilde{K}_2(t),\tilde{\Upsilon}_1(t,\cdot),\tilde{\Upsilon}_2(t,\cdot) ,\pi(t),\theta_0(t),\theta_1(t),\theta_2(t,\cdot))| &<& \infty\,,
 \end{eqnarray*}
 which contradicts the assumption. Then $\mathcal{H}$ is Lipschitz continuous.
\end{proof}

Then, following Theorem \ref{theorembsde} above, we establish the relationship between the
value function of the game problem and the solution of a delayed BSDE with jumps. Thus,
the value function $\tilde{\mathcal{J}}(t,x)$ is given by the following noisy memory BSDE:
\begin{eqnarray}\nonumber
  && d\tilde{\mathcal{J}}(t) \\ \nonumber
  &=& -\mathcal{H}(t,X(t),Y(t),U(t), K_1(t),K_2(t),\Upsilon_1(t,\cdot),
\Upsilon_2(t,\cdot),\pi^*(t),\theta_0^*(t),\theta_1^*(t),\theta_2^*(t,\cdot))dt \\ \nonumber
   && +K_1(t)d\widehat{W}(t) +K_2(t)dW_1(t)
+\int_{\mathbb{R}}\Upsilon_1(t,z)\widehat{N}_{\Lambda}(dt,dz)
+\int_{\mathbb{R}}\Upsilon_2(t,z)\widehat{N}^0_{\Lambda}(dt,dz)\,,
\end{eqnarray}
with the terminal condition $\tilde{\mathcal{J}}(T) = h(X(T),Y(T))$.

In fact, the existence and uniqueness of the solution to the above delayed BSDE with jumps is guaranteed from the Lipschitz condition proved in Lemma \ref{lemma}. Then the solution of the delayed BSDE is given by
\begin{eqnarray*}
 \tilde{\mathcal{J}}(t)  &=& \mathbb{E}\Bigl[h(X(T),Y(T))-\int_t^T\tilde{\ell}(s,X(s),Y(s),U(s),\pi^*(s),\theta_0^*(s),\theta_1^*(s),\theta_2^*(s,\cdot))ds \mid\mathcal{F}_t\Bigl] \\
   &=&  \mathcal{V}(\pi^*,\theta_1^*,\theta_2^*,\theta_3^*)\,,
\end{eqnarray*}
which is the optimal value function from Theorem \ref{theorembsde}.

\section{A quadratic penalty function case}

In this section, we consider a convex risk measure with quadratic penalty. We
derive explicit solutions when $\ell$ is quadratic in $\theta_0,\,\theta_1,\,\theta_2$
and identical zero in $h$. The penalty function under consideration here, may be
related to the entropic penalty function considered, for instance, by Delbaen
{\it et. al.} \cite{delbaen}. It has also been adopted by Elliott and Siu
\cite{Elliott2011}, Siu \cite{siu2012} and Meng and Siu \cite{Meng}. We obtain
the explicit optimal investment strategy and the optimal risks for this case of
a risk-based optimization problem with jumps, regime switching and noisy delay. Finally, we consider some particular cases and we see using some numerical parameters, how an insurer can allocate his portfolio.

Suppose that the penalty function is given by
\begin{eqnarray*}
&& \ell(t,X(t),Y(t),Z(t),\pi(t),\theta_0(t),\theta_1(t),\theta_2(t,\cdot)) \\
&:=& \frac{1}{2(1-\delta)}
\Bigl(\theta_0^2(t)+ \theta_1^2(t)+ \int_{\mathbb{R}}\theta_2^2(t,z)\nu_{\hat{\Lambda}}(dz)\Bigl)
G^{\theta_0,\theta_1,\theta_2}(t)\,,
\end{eqnarray*}
where $1-\delta$ is a measure of an insurer's relative risk aversion and
$\delta<1$. Then, the Hamiltonian $\mathcal{H}$ becomes:
\begin{eqnarray*}
  && \mathcal{H}(t,X(t),Y(t),U(t), K_1(t),K_2(t),\Upsilon_1(t,\cdot),
\Upsilon_2(t,\cdot),\pi(t),\theta_0(t),\theta_1(t),\theta_2(t,\cdot)) \\
  &=&
-G^{\theta_0,\theta_1,\theta_2}(t)\Bigl[p(t)+(r(t)-\vartheta(t)-\xi)X(t)+\pi(t)(\hat{
\alpha}^{\Lambda}(t)-r(t)) +(\bar{\vartheta}(t)-\kappa\zeta)Y(t)  \\
   && -\xi U(t) +\pi(t)\beta(t)\theta_0(t) +\theta_1(t)\kappa
X(t)(1-e^{-\zeta\varrho}\chi_{[0,T-\varrho]}) \\
   && +\sum_{j=1}^D\langle\hat{\Lambda}(t-),e_j\rangle
\Bigl(\pi(t)\int_{\mathbb{R}}z\theta_2(t,z)\varepsilon_j(t)\nu_j(dz)
-\int_0^\infty\lambda_j(t)z(1+\theta_0(t))f_j(dz)\Bigl)\Bigl] \\
   &&  -
\frac{1}{2(1-\delta)}\Bigl(\theta_0^2(t)+ \theta_1^2(t)+ \int_{\mathbb{R}}\theta_2^2(t,z)\nu_{\hat{\Lambda}}(dz)\Bigl)
G^{\theta_0,\theta_1,\theta_2}(t)\,.
\end{eqnarray*}
Applying the first order condition for maximizing the Hamiltonian with respect
to $\theta_0,\,\theta_1$ and $\theta_2$, and minimizing with respect to $\pi$, we obtain
the following

\begin{eqnarray*}
\pi^*(t) &=& \frac{\hat{\alpha}^{\Lambda}(t)-r(t)+(1-\delta)\beta(t)\Bigl(\sum_{j=1}
^D\langle\hat{\Lambda}(t-),e_j\rangle \int_0^{\infty}z\lambda_j(t)f_j(dz)
\Bigl)}{(1-\delta) \Bigl(\beta^2(t)+\sum_{j=1}^D\langle\hat{\Lambda}(t-),e_j\rangle
\int_{\mathbb{R}}z^2\varepsilon_j(t)\nu_j(dz)\Bigl)}\,, \\
 \theta_0^*(t) &=&
(1-\delta)\Bigl[\sum_{j=1}^D\langle\hat{\Lambda}(t-),e_j\rangle
\int_0^{\infty}z\lambda_j(t)f_j(dz) \\
&& -\frac{\hat{\alpha}^{\Lambda}(t)-r(t)+(1-\delta)\beta(t)\Bigl(\sum_{j=1}
^D\langle\hat{\Lambda}(t-),e_j\rangle \int_0^{\infty}z\lambda_j(t)f_j(dz)
\Bigl)}{(1-\delta) \Bigl(\beta^2(t)+
\sum_{j=1}^D\langle\hat{\Lambda}(t-),e_j\rangle
\int_{\mathbb{R}}z^2\varepsilon_j(t)\nu_j(dz)\Bigl)}\beta(t)\Bigl]\,, \\
\theta_1^*(t) &=& (\delta-1)\kappa X(t)(1-e^{-\zeta\varrho}\chi_{[0,T-\varrho]})
\end{eqnarray*}
and
$$
\theta_2^*(t,z)=(\delta-1)z\frac{\hat{\alpha}^{\Lambda}(t)-r(t)+(1-\delta)\beta(t)\Bigl(\sum_{j=1}
^D\langle\hat{\Lambda}(t-),e_j\rangle \int_0^{\infty}z\lambda_j(t)f_j(dz)
\Bigl)}{(1-\delta) \Bigl(\beta^2(t)+
\sum_{j=1}^D\langle\hat{\Lambda}(t-),e_j\rangle
\int_{\mathbb{R}}z^2\varepsilon_j(t)\nu_j(dz)\Bigl)}\,.
$$

Then, the value function of the game problem is given by the following BSDE:
\begin{eqnarray*}
d\mathcal{J}(t) &=&
G^{\theta_0^*,\theta_1^*,\theta_2^*}(t)\Bigl[p(t)+(r(t)-\vartheta(t)-\xi)X(t)+\pi^*(t)(\hat{
\alpha}^{\Lambda}(t)-r(t)) +(\bar{\vartheta}(t)-\kappa\zeta)Y(t)  \\
   && -\xi U(t) +\pi^*(t)\beta(t)\theta_0^*(t) +\theta_1^*(t)\kappa
X(t)(1-e^{-\zeta\varrho}\chi_{[0,T-\varrho]}) \\
   && +\sum_{j=1}^D\langle\hat{\Lambda}(t-),e_j\rangle
\Bigl(\pi^*(t)\int_{\mathbb{R}}z\theta_2^*(t,z)\varepsilon_j(t)\nu_j(dz)
-\int_0^\infty\lambda_j(t)z(1+\theta_0^*(t))f_j(dz)\Bigl)\Bigl] \\
   &&  +
\frac{1}{2(1-\delta)}\Bigl((\theta_0^*)^2(t)+ (\theta_1^*)^2(t)+ \int_{\mathbb{R}}(\theta_2^*)^2(t,z)\nu_{\hat{\Lambda}}(dz)\Bigl)\Bigl]dt \\
   && +K_1(t)d\widehat{W}(t) +K_2(t)dW_1(t)
+\int_{\mathbb{R}}\Upsilon_1(t,z)\widehat{N}_{\Lambda}(dt,dz)
+\int_{\mathbb{R}}\Upsilon_2(t,z)\widehat{N}^0_{\Lambda}(dt,dz)\,.
\end{eqnarray*}

\begin{example}
Suppose that the the driving processes $\tilde{N}$ and $\tilde{N}^0$ are Poisson processes $N$ and $N^0$, with the jump intensities $\lambda$ and $\lambda^0$. We consider the following cases:
\begin{itemize}
  \item[Case 1.] We suppose that there is no regime switching in the model, then the optimal investment strategy is given by
  $$
  \pi^*(t)=\frac{\alpha(t)-r(t)}{(1-\delta)(\beta^2(t)+\lambda)}\,.
  $$
  We assume the following hypothetical parameters: interest rate $r=4.5\%$, the appreciation rate $\alpha=11\%$, the volatility $\beta=20\%$, the insurer's relative risk aversion $\delta=0.5$ and the jump intensity given by $\lambda=0.5$. Then the optimal portfolio invested in the risky asset is given by $\pi^*=0.24074$, i.e., $24.074\%$ of the wealth should be invested in the risky share.
  \item[Case 2.] We suppose existence of two state Markov chain $\mathcal{S}=\{e_1,e_2\}$, where the states $e_1$ and $e_2$ represent the expansion and recession of the economy respectively. By definition, $\langle\hat{\Lambda}(t),e_1\rangle=\mathbb{P}(X(t)=e_1\mid\mathcal{G}_t)$ and $\langle\hat{\Lambda}(t),e_2\rangle=1-\mathbb{P}(X(t)=e_1\mid\mathcal{G}_t)$. Let $\alpha_i,r_i,\lambda_i,\lambda_i^0$ be the associate parameters when the economy is in state $e_i$, $i=1,2$. Then the optimal portfolio is given by
      \begin{eqnarray*}
        \pi^*(t) &=& \frac{[\alpha_1(t)-r_1(t)-(\alpha_2(t)-r_2(t))+(1-\delta)\beta(t) (\lambda_1^0(t)-\lambda_2^0(t))] \mathbb{P}(X(t)=e_1\mid\mathcal{G}_t)}{(1-\delta)[\beta^2(t)+\lambda_2(t)+(\lambda_1(t)-\lambda_2(t)) \mathbb{P}(X(t)=e_1\mid\mathcal{G}_t)]}  \\
         && + \frac{\alpha_2(t)-r_2(t)+(1-\delta)\beta(t) \lambda_2^0}{(1-\delta)[\beta^2(t)+\lambda_2(t)+(\lambda_1(t)-\lambda_2(t)) \mathbb{P}(X(t)=e_1\mid\mathcal{G}_t)]}\,.
      \end{eqnarray*}
      In this case, we consider the following parameters: $\alpha_1=13\%,\,\alpha_2=9\%,\,r_1=4.5\%,\,r_2=9\%,\,\beta=20\%,\, \lambda_1^0=\lambda_1=0.5,\,\lambda_2=\lambda_2^0=0.7,\,\delta=0.5$ and $\mathbb{P}(X=e_1)=70\%$. Then $\pi^*=0.28$, i.e., $28\%$ of the wealth should be invested in the risky share.
\end{itemize}

\end{example}

\subsection*{Acknowledgment}

We would like to express our deep gratitude to the University of Pretoria and
the MCTESTP Mozambique for their support.

\end{document}